\documentclass[12pt,draftcls,onecolumn]{IEEEtran}
\usepackage{amsmath,amssymb,amsthm}    
\usepackage[dvips]{graphicx}   
\usepackage{verbatim}   
\usepackage{color}      
\usepackage{cite}
\usepackage{epsfig}
\usepackage{float}
\usepackage{algorithm}
\usepackage{algorithmic}

\newtheorem{lemma}{\bf Lemma}
\newtheorem{theorem}{\bf Theorem}

\newcounter{mytempeqncnt}

\makeatletter
\newcommand*{\rom}[1]{\expandafter\@slowromancap\romannumeral #1@}
\makeatother
\begin{document}
\title{\textbf{Distributed Detection in Tree Topologies with Byzantines}}
\author{Bhavya~Kailkhura,~\IEEEmembership{Student Member,~IEEE}, Swastik~Brahma,~\IEEEmembership{Member,~IEEE}, Yunghsiang~S. Han,~\IEEEmembership{Fellow,~IEEE}, Pramod~K.~Varshney,~\IEEEmembership{Fellow,~IEEE}
\thanks{Some related preliminary work was presented at the International Conference on Computing, Networking and Communications Workshops (ICNC-2013), San Diego,
CA, January 2013.}
\thanks{ 
B. Kailkhura, S. Brahma and P. K. Varshney are with Department of EECS, Syracuse University, Syracuse, NY 13244. (email: bkailkhu@syr.edu; skbrahma@syr.edu; varshney@syr.edu)}
\thanks{Y. S. Han is with EE Department, National Taiwan University of Science and Technology, Taiwan, R. O. C. (email: yshan@mail.ntust.edu.tw)}}
\date{}
\maketitle

\begin{abstract}
In this paper, we consider the problem of distributed
detection in tree topologies in the presence of Byzantines.
The expression for minimum attacking power required by the Byzantines to blind the fusion center (FC) is obtained. More specifically, we show that when more than a certain fraction of individual node decisions are falsified, the decision fusion scheme becomes completely incapable. We obtain closed form expressions for the optimal attacking strategies that minimize the detection error exponent at the FC. We also look at the possible counter-measures from the FC's perspective to protect the network from these Byzantines. We formulate the robust topology design problem as a bi-level program and provide an efficient algorithm to solve it. We also provide some numerical results to gain insights
into the solution.
\end{abstract}
\section{Introduction}
Distributed detection  has been a well studied topic in the detection theory 
literature \cite{Varshney}\cite{Viswanathan}\cite{veer} and has traditionally focused 
on the parallel network topology. In distributed detection with parallel topology,
nodes make their local decisions regarding the underlying
phenomenon and send them to the fusion center (FC), where
a global decision is made.
Even though the parallel topology has received significant attention, there are many practical situations where parallel topology cannot be implemented due to several factors, such as, the FC being outside the communication range of the nodes and limited energy budget of the nodes~\cite{Lin}. 
In such cases, a multi-hop network is employed, where nodes are organized hierarchically into multiple levels (tree networks).
With intelligent use of resources across levels, tree networks have the potential to provide 
a suitable balance between cost, coverage, functionality, and reliability \cite{purush}.
Some examples of tree networks include wireless sensor and military communication networks.
For instance, the IEEE 802.15.4 (Zigbee) specifications \cite{Aliance} and IEEE 802.22b \cite{802_22bDraft}
can support tree-based topologies. Theses nodes are often deployed in open and unattended
environments and are vulnerable to physical tampering.

In recent years, security issues of distributed inference networks are increasingly being studied.
One typical attack on such networks is a Byzantine attack. While Byzantine attacks 
(originally proposed by \cite{Lamport}) may, in general, refer to many types of malicious
behavior; our focus in this paper is on data-falsification attacks ~\cite{avemp,frag, Rifa, Marano, Rawat, Kailkhura2013, Kailkhura, aditya,a2}. In this type of attack, the compromised node may send false (erroneous) local decisions to the FC 
to degrade the detection performance. This attack becomes more severe in tree topologies where malicious nodes 
can alter local decisions of a large part of the network and cause degradation of system 
performance and may even make the decision fusion schemes to become completely incapable. In this paper, we refer to such a data falsification attacker as a Byzantine. 
\subsection{Related Work}
Although distributed detection has been a very active field of
research in the past~\cite{Varshney,Viswanathan,veer}, security problems in distributed detection
networks gained attention only very recently. In \cite{Marano}, the authors considered the problem of distributed detection in the presence of Byzantines for a  parallel topology and determined the optimal attacking strategy which minimizes the detection error exponent. They assumed that the Byzantines know the true hypothesis, which obviously is not satisfied in practice but does provide a bound. 
In \cite{Rawat}, the authors analyzed the same problem in the context of collaborative spectrum sensing. They relaxed the  assumption of perfect knowledge of the hypotheses by assuming that the Byzantines obtain knowledge about the true hypotheses from their own sensing observations.

The above work~\cite{Marano,Rawat} addresses the issue of Byzantines from the attacker's perspective. Schemes to mitigate the effect of Byzantines have also been proposed in the literature. In~\cite{Rawat}, the authors proposed a simple scheme to identify the Byzantines. The idea was to maintain a reputation metric for every node by comparing each node's local decision to the global decision made at the FC using the majority rule. 
In~\cite{aditya}, the authors proposed another scheme to mitigate the effect of Byzantines in a parallel topology. The idea behind the proposed identification scheme is to compare every node's observed behavior over time with the expected behavior of an honest node. The nodes whose observed behavior is sufficiently far from the expected behavior are tagged as Byzantines and this information is employed while making a decision at the FC. 
In \cite{a2}, the authors investigated the problem of distributed detection in the presence of different types of
Byzantine nodes. Each Byzantine type corresponds to a different operating point and, therefore, the problem of identifying different Byzantine nodes along with their operating points was considered. Once the Byzantine operating points are estimated, this information was utilized by the FC to improve global detection performance.
The problem of designing the optimal fusion rule and the local sensor thresholds with Byzantines for a parallel topology was considered in~\cite{Kailkhura}. 
\subsection{Main Contributions}
All the approaches discussed so far consider distributed detection with Byzantines for parallel topologies. In contrast to previous work, we study the problem of distributed detection with Byzantines for tree topologies. More specifically, we address the
problem of distributed detection in perfect $a$-ary tree networks\footnote{For previous works on perfect $a$-ary tree networks, please see~\cite{Gure},~\cite{sab},~\cite{Pad}.} in the presence of Byzantine attacks (data falsification attcks). We assume that the cost of attacking nodes at different levels is different and analyze the problem under this assumption. In our preliminary work on this problem~\cite{Kailkhura2013}, we analyzed the  problem only from an attacker's perspective assuming that the honest and Byzantine nodes are identical in terms of their
detection performance. In our current work, we significantly extend our previous work and investigate the problem from both the attacker's and the FC's perspective. For the analysis of the optimal attack, we allow Byzantines to have different detection performance than the honest nodes and, therefore, provide a more general and comprehensive analysis of the problem compared to our previous work~\cite{Kailkhura2013}.
The main contributions of this paper are as follows.
\begin{itemize}
\item We obtain a closed form expression for the minimum attacking power required by the Byzantines to blind the FC in a tree network and show that when more than a certain fraction of individual node decisions are falsified, the decision fusion scheme becomes completely incapable.
\item When the fraction of Byzantines is not
sufficient to blind the FC, we provide closed form expressions for the optimal attacking strategies for the Byzantines that most degrade the detection performance.
\item We also look at the problem from the network designer's (FC) perspective. More specifically, we formulate the robust tree topology design problem as a bi-level program and provide an efficient algorithm to solve it, which is guaranteed to find an optimal solution, if one exists. 
\end{itemize}

The rest of the paper is organized as follows.
Section~\ref{sec2} introduces our system model.
In Section~\ref{sec3}, we study the problem from Byzantine's perspective and provide closed form expressions for optimal attacking strategies.
In Section~\ref{sec4}, we formulate the robust topology design problem as a bi-level program and provide an efficient algorithm to solve it in polynomial time. 
Finally, Section~\ref{sec6} concludes the paper.
\section{System Model}
\label{sec2}
\begin{figure}[t]
  \centering
    \includegraphics[height=2.5in, width=!]{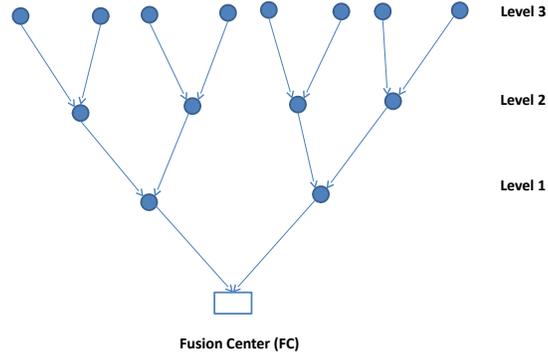}
    \vspace*{-0.3in}
    \caption{A distributed detection system organized as a perfect binary tree $T(3,\;2)$ is shown as an example.}\label{syst}
\vspace*{-0.1in}
\end{figure}
We consider a distributed detection system with the topology of a perfect $a$-tree $T(K,\;a)$ rooted at the FC (See Fig.~\ref{syst}). A perfect $a$-tree is an $a$-ary tree in which all the leaf nodes are at the same depth and all the internal nodes have degree `$a$'. $T(K,\;a)$ has a set $\mathcal{N}=\{\mathbb{N}_{k}\}_{k=1}^{K}$ of transceiver nodes, where $|\mathbb{N}_{k}|=N_{k}=a^k$ is the total number of nodes at level (or depth) $k$.
We assume that the depth of the tree is $K>1$ and the number of children is $a\geq 2$. The total number of nodes in the network is denoted as $\sum_{k=1}^K N_{k}=N$. $\mathcal{B}=\{\mathbb{B}_{k}\}_{k=1}^{K}$ denotes the set of Byzantine nodes with $|\mathbb{B}_{k}|=B_{k}$, where $\mathbb{B}_{k}$ is the set of Byzantines at level $k$.
We assume that the FC is not aware of the exact set of Byzantine nodes and considers each node at level $k$ to
be Byzantine with a certain probability $\alpha_{k}$. 
In practice, nodes operate with very limited energy and, therefore, it
is reasonable to assume that the packet IDs (or source IDs) are not forwarded in the tree to save
energy. Moreover, even in cases where the packet IDs (or source IDs) are forwarded, notice that the packet IDs (or source IDs) can be tempered too, thereby preventing the FC to be deterministically aware of the source of a message. Therefore, we consider that the FC looks at messages coming from nodes in a probabilistic manner and considers each received bit to originate from nodes at level $k$ with certain probability $\beta_{k}\in [0,1]$.
This also implies that, from the FC's perspective, received bits are identically distributed. For a $T(K,\;a)$, 
$$\beta_{k}=\frac{a^k}{N}.$$ 

\subsection{Distributed detection in a tree topology}
\label{Network}

We consider a binary hypothesis testing problem with the two hypotheses $H_{0}$ (signal is absent)
and $H_{1}$ (signal is present).
Each node $i$ at level $k$ acts as a source in that it makes a one-bit local decision $v_{k,i} \in \{ 0,1\}$
and sends $u_{k,i}$ to its parent node at level $k-1$, where $u_{k,i} = v_{k,i}$
if $i$ is an uncompromised (honest) node, but for a compromised (Byzantine) node $i$,
$u_{k,i}$ need not be equal to $v_{k,i}$. 
It also receives the decisions $u_{k',j}$ of all successors $j$
at levels $k' \in [k+1,K]$, which are forwarded to $i$ by its immediate children. It 
forwards\footnote{For example, IEEE 802.16j mandates tree forwarding and IEEE 802.11s standardizes a tree-based routing protocol.} these received decisions along with $u_{k,i}$ to its parent node at level $k-1$. If node $i$ is a Byzantine, then it might alter 
these received decisions before forwarding.
We assume error-free communication channels between children and the parent nodes.
We denote the probabilities of detection and false alarm of a honest node $i$ at level $k$ by $P_{d}^H=P(v_{k,i}=1|H_{1},i\notin \mathbb{B}_{k})$ and $P_{fa}^H=P(v_{k,i}=1|H_{0},i\notin \mathbb{B}_{k})$, respectively. Similarly, the probabilities of detection and false alarm of a Byzantine node $i$ at level $k$ are denoted by $P_{d}^B=P(v_{k,i}=1|H_{1},i\in \mathbb{B}_{k})$ and $P_{fa}^B=P(v_{k,i}=1|H_{0},i\in \mathbb{B}_{k})$, respectively. 

\subsection{Byzantine attack model}
\label{model}
Now a mathematical model for the Byzantine attack is presented.
If a node 
is honest, then it forwards its own decision and received decisions without altering them.
However, a
Byzantine node, in order to undermine the network performance, may alter its decision as well as received decisions from its children prior to transmission.
We define the following strategies $P_{j,1}^H$, $P_{j,0}^H$ and $P_{j,1}^B$, $P_{j,0}^B$ ($j \in \{0,1\}$)
for the honest and Byzantine nodes, respectively:\\
Honest nodes:
\begin{equation}
P_{1,1}^H=1-P_{0,1}^H=P^{H}(x=1|y=1)=1\label{P11H}
\end{equation}
\begin{equation}
P_{1,0}^H=1-P_{0,0}^H=P^{H}(x=1|y=0)=0\label{P10H}
\end{equation}
				
\noindent
Byzantine nodes:
\begin{equation}
P_{1,1}^B=1-P_{0,1}^B=P^{B}(x=1|y=1)\label{P11B}
\end{equation}
\begin{equation}
P_{1,0}^B=1-P_{0,0}^B=P^{B}(x=1|y=0)\label{P10B}
\end{equation}
where $P(x=a|y=b)$ is the probability that a node sends $a$ to its parent when it receives $b$ from its child or its actual decision is $b$.
Furthermore, we assume that if a node (at any level) is a Byzantine then none of its ancestors  are Byzantines; otherwise, the effect of a Byzantine due to other Byzantines on the same path may be nullified  (e.g., Byzantine ancestor re-flipping the already flipped decisions of its successor).
 This means that any path from a leaf node to the FC will have at most one Byzantine. Thus, we have, 
$\sum_{k=1}^{K}\alpha_{k}\leq 1$ since the average number of Byzantines along any path from a leaf to the root cannot be greater than $1$.

\subsection{Performance metric}
\label{sec:metric}
The Byzantine attacker always wants to degrade the detection performance at the FC as much as possible; in contrast, the FC wants to maximize the detection performance. In this work, we employ the Kullback-Leibler divergence (KLD) \cite{Kullback} to be the network performance metric that characterizes detection performance. The KLD is a frequently used information-theoretic “distance” measure to characterize  detection performance. By Stein's lemma, we know that in the Neyman-Pearson setup for a fixed missed detection probability, the false alarm probability obeys the asymptotics
\begin{equation}
\lim_{N \rightarrow \infty} \frac{\ln P_F}{N}=-D,\;\text{for a fixed}\; P_M,
\end{equation}
where $P_M$, $P_F$ are missed detection and false alarm probabilities, respectively. 
The KLD between the distributions $\pi_{j,0}=P(z=j|H_{0})$ and $\pi_{j,1}=P(z=j|H_{1})$ can be expressed as
\begin{equation}
\label{klde}
D(\pi_{j,1}||\pi_{j,0})=\sum_{j \in \{0,1\}}P(z=j|H_{1})\log \dfrac{P(z=j|H_{1})}{P(z=j|H_{0})}.
\end{equation}

\begin{figure*}[t]
\normalsize
\setcounter{mytempeqncnt}{\value{equation}}

\begin{eqnarray}
P(z_{i}=j|H_{0})&=&\left[\sum_{k=1}^{K}\beta_{k}\left(\sum_{i=1}^{k}\alpha_{i}\right)\right][P_{j,0}^{B}(1-P_{fa}^B)+P_{j,1}^{B}P_{fa}^B]\nonumber\\
&+&\left[\sum_{k=1}^{K}\beta_{k}\left(1-\sum_{i=1}^{k}\alpha_{i}\right)\right][P_{j,0}^{H}(1-P_{fa}^H)+P_{j,1}^{H}P_{fa}^H]
\label{eq1}
\end{eqnarray}
\vspace*{-0.1in}
\begin{eqnarray}
P(z_{i}=j|H_{1})&=&\left[\sum_{k=1}^{K}\beta_{k}\left(\sum_{i=1}^{k}\alpha_{i}\right)\right][P_{j,0}^{B}(1-P_{d}^B)+P_{j,1}^{B}P_{d}^B]\nonumber\\
&+&\left[\sum_{k=1}^{K}\beta_{k}\left(1-\sum_{i=1}^{k}\alpha_{i}\right)\right][P_{j,0}^{H}(1-P_{d}^H)+P_{j,1}^{H}P_{d}^H]
\label{eq2}
\end{eqnarray}
\hrulefill
\vspace*{-0.23in}
\end{figure*}
For a $K$-level network, distributions of received decisions at the FC
$z_{i}$, $i=1,..,N$, under $H_{0}$ and $H_{1}$ are given by \eqref{eq1} and \eqref{eq2}, respectively.
In order to make the analysis tractable, we assume that the  network designer attempts to maximize the KLD of each node as seen by the FC. On the other hand, the attacker attempts to minimize the KLD of each node as seen by the FC.

Next, we explore the optimal attacking strategies for the Byzantines that most degrade the detection performance by minimizing KLD.
\section{Optimal Byzantine Attack}
\label{sec3}
As discussed earlier, the Byzantine nodes attempt to make their KL divergence as small as possible. Since the KLD is always non-negative,
Byzantines attempt to choose $P(z=j|H_{0})$ and $P(z=j|H_{1})$ such that KLD is zero. In this case, an adversary can make the data that the FC receives from the nodes such that no information is conveyed. This is possible when
\begin{equation}
P(z=j|H_{0})=P(z=j|H_{1})\qquad \forall j\in \{0,1\}.
\label{eq8}
\end{equation}
Substituting \eqref{eq1} and \eqref{eq2} in \eqref{eq8} and after simplification, the condition to make the $KLD=0$ for a $K$-level network can be expressed as
\begin{equation}
P_{j,1}^{B}-P_{j,0}^{B}=\frac{\sum_{k=1}^{K}[\beta_{k}(1-\sum_{i=1}^{k}\alpha_{i})]}{\sum_{k=1}^{K}[\beta_{k}(\sum_{i=1}^{k}\alpha_{i})]]}\dfrac{P_{d}^H-P_{fa}^H}{P_{d}^B-P_{fa}^B}(P_{j,0}^{H}-P_{j,1}^{H}).
\end{equation}
From \eqref{P11H} to \eqref{P10B}, we have
\begin{equation}
P_{0,1}^{B}-P_{0,0}^{B}=\frac{\sum_{k=1}^{K}[\beta_{k}(1-\sum_{i=1}^{k}\alpha_{i})]}{\sum_{k=1}^{K}[\beta_{k}(\sum_{i=1}^{k}\alpha_{i})]]}\dfrac{P_{d}^H-P_{fa}^H}{P_{d}^B-P_{fa}^B}=-(P_{1,1}^{B}-P_{1,0}^{B}).
\end{equation}
Hence, the attacker can degrade detection performance by intelligently choosing $(P_{0,1}^{B}, P_{1,0}^{B})$, which are dependent on $\alpha_{k}$, for $k=1, \cdots, K$. Observe that, 
$$0\le P_{0,1}^{B}-P_{0,0}^{B}$$ since $\sum_{i=1}^{k}\alpha_{i}\le 1$ for $k\le K$. To make $KLD=0$, we must have
$$P_{0,1}^{B}-P_{0,0}^{B}\le 1$$
such that $(P_{j,1}^B,P_{j,0}^B)$ becomes a valid probability mass function. Notice that, when $P_{0,1}^{B}-P_{0,0}^{B}>1$ there does not exist any attacking probability distribution $(P_{j,1}^B,P_{j,0}^B)$ that can make $KLD = 0$. In the case of $P_{0,1}^{B}-P_{0,0}^{B}=1$, there exists a unique solution $(P_{1,1}^B,P_{1,0}^B)=(0,1)$ that can make $KLD = 0$. For the $P_{0,1}^{B}-P_{0,0}^{B}<1$ case, there exist an infinite number of attacking probability distributions $(P_{j,1}^B,P_{j,0}^B)$ which can make $KLD = 0$.

By further assuming that the honest and Byzantine nodes are identical in terms of their detection performance, i.e., $P_{d}^H=P_{d}^B$ and $P_{fa}^H=P_{fa}^B$,  the above condition to blind the FC reduces to
$$\frac{\sum_{k=1}^{K}[\beta_{k}(1-\sum_{i=1}^{k}\alpha_{i})]}{\sum_{k=1}^{K}[\beta_{k}(\sum_{i=1}^{k}\alpha_{i})]]}\le 1$$
which is equivalent to
\begin{equation}
\sum_{k=1}^{K}[\beta_{k}(1-2(\sum_{i=1}^{k}\alpha_{i}))]\leq 0.\label{alpha-condition}
\end{equation}

Recall that  $\alpha_{k}=\frac{B_{k}}{N_{k}}$ and  $\beta_{k}=\frac{N_{k}}{\sum_{i=1}^{K}N_{i}}$. Substituting $\alpha_{k}$ and $\beta_{k}$  into \eqref{alpha-condition} and simplifying the result,  we have the following theorem.

\begin{theorem}
\label{th1}
In a tree network with $K$ 
levels, there exists an attacking probability distribution $(P_{0,1}^{B}, P_{1,0}^{B})$ that can make $KLD=0$, and
thereby blind the FC, if and only if $\{B_{k}\}_{k=1}^{K}$ satisfy
\begin{equation}
\label{fifty}
\sum_{k=1}^K \left(\frac{B_{k}}{N_{k}}\sum_{i=k}^K N_{i}\right)\geq \frac{N}{2}. 
\end{equation}
\end{theorem}

Dividing both sides of~\eqref{fifty} by $N$, the above condition can be written as $\sum_{k=1}^{K}\beta_{k} \sum_{i=1}^{k}\alpha_{i}\geq 0.5$. 
This implies that to make the FC blind,
$50\% $ or more nodes in the network need to be covered\footnote{Node $i$ at level $k'$ covers all its children at levels $k'+1\leq k\leq K$ and the node $i$ itself and, therefore, the total number of covered nodes by $B_{k'}$, Byzantine at level $k'$, is $\dfrac{B_{k'}}{N_{k'}}.\sum_{i=k'}^K N_{i}$.} by
the Byzantines.
Next, to explore the optimal attacking probability distribution $(P_{0,1}^{B}, P_{1,0}^{B})$ that minimizes $KLD$ when \eqref{alpha-condition} does not hold, we explore the properties of KLD. 

First, we show that attacking with symmetric flipping probabilities is the optimal strategy in the region where the attacker cannot blind the FC. In other words, attacking with $P_{1,0}=P_{0,1}$ is the optimal strategy for the Byzantines. For analytical tractability, we assume
$P_{d}^H=P_{d}^B=P_d$ and $P_{fa}^H=P_{fa}^B=P_{fa}$ in further analysis. 
\begin{lemma}
\label{lem}
In the region where the attacker cannot blind the FC, the optimal attacking strategy comprises of symmetric flipping probabilities. More specifically, any non zero deviation $\epsilon_i\in(0,p]$ in flipping probabilities $(P_{0,1}^{B}, P_{1,0}^{B})=(p-\epsilon_1,p-\epsilon_2)$, where $\epsilon_1\neq\epsilon_2$, will result in increase in the KLD.
\end{lemma}
\begin{proof}
Let us denote, $P(z=1|H_1)=\pi_{1,1}$, $P(z=1|H_0)=\pi_{1,0}$ and $t=\sum_{k=1}^{K}\beta_{k} \sum_{i=1}^{k}\alpha_{i}$. Notice that, in the region where the attacker cannot blind the FC, the parameter $t<0.5$. To prove the lemma, we first show that any positive deviation $\epsilon\in(0,p]$ in flipping probabilities $(P_{1,0}^{B}, P_{0,1}^{B})=(p,p-\epsilon)$ will result in an increase in the KLD. After plugging in $(P_{1,0}^{B}, P_{0,1}^{B})=(p,p-\epsilon)$ in \eqref{eq1} and \eqref{eq2}, we get
\begin{eqnarray}
\pi_{1,1}&=&t(p-P_d(2p-\epsilon))+P_d\\\label{eeq1}
\pi_{1,0}&=&t(p-P_{fa}(2p-\epsilon))+P_{fa}.\label{eeq2}
\end{eqnarray}
Now we show that the KLD, $D$, as give in \eqref{klde} is a monotonically increasing function of the parameter $\epsilon$ or in other words, $\dfrac{dD}{d\epsilon}>0$.
\begin{eqnarray}
\dfrac{dD}{d\epsilon}&=&
\pi_{1,1}\left(\dfrac{\pi_{1,1}'}{\pi_{1,1}}-\dfrac{\pi_{1,0}'}{\pi_{1,0}}\right)+\pi_{1,1}' \log \dfrac{\pi_{1,1}}{\pi_{1,0}}\nonumber\\
&+&(1-\pi_{1,1})\left(\dfrac{\pi_{1,0}'}{1-\pi_{1,0}}-\dfrac{\pi_{1,1}'}{1-\pi_{1,1}}\right)
-\pi_{1,1}'\log \dfrac{1-\pi_{1,1}}{1-\pi_{1,0}}\label{con1}
\end{eqnarray}
where $\dfrac{d\pi_{1,1}}{d\epsilon}=\pi_{1,1}'=tP_d$ and $\dfrac{d\pi_{1,0}}{d\epsilon}=\pi_{1,0}'=tP_{fa}$ and $t$ is the fraction of covered nodes by the Byzantines. After rearranging the terms in the above equation, the condition $\dfrac{dD}{d\epsilon}>0$ becomes 
\begin{eqnarray}
\dfrac{1-\pi_{1,1}}{1-\pi_{1,0}}+\dfrac{P_d}{P_{fa}}\log\dfrac{\pi_{1,1}}{\pi_{1,0}}>\dfrac{\pi_{1,1}}{\pi_{1,0}}+\dfrac{P_d}{P_{fa}}\log\dfrac{1-\pi_{1,1}}{1-\pi_{1,0}}.\label{log}
\end{eqnarray}
Since $P_d>P_{fa}$ and $t<0.5$, $\pi_{1,1}>\pi_{1,0}$. It can also be proved that $\dfrac{P_{fa}}{P_{d}}\dfrac{\pi_{1,1}}{\pi_{1,0}}<1$. Hence, we have
\begin{equation*}
1+(\pi_{1,1}-\pi_{1,0})>\dfrac{P_{fa}}{P_{d}}\dfrac{\pi_{1,1}}{\pi_{1,0}}
\end{equation*}
which is equivalent to
\begin{equation}
\dfrac{1-\pi_{1,1}}{1-\pi_{1,0}}+\dfrac{P_d}{P_{fa}}
\left(1-\dfrac{\pi_{1,0}}{\pi_{1,1}}\right)
>\dfrac{\pi_{1,1}}{\pi_{1,0}}+\dfrac{P_d}{P_{fa}}
\left(\dfrac{1-\pi_{1,1}}{1-\pi_{1,0}}-1\right).\label{eq-1}
\end{equation}
Applying the logarithm inequality $(x-1)\geq\log x \geq \dfrac{x-1}{x}$, for $x>0$ to \eqref{eq-1}, one can prove that condition \eqref{log} is true.

Similarly, we can show that any non zero deviation $\epsilon\in(0,p]$ in flipping probabilities $(P_{1,0}^{B}, P_{0,1}^{B})=(p-\epsilon,p)$ will result in an increase in the KLD, i.e.,
$\dfrac{dD}{d\epsilon}>0$,
 or
\begin{eqnarray}
\dfrac{\pi_{1,1}}{\pi_{1,0}}+\dfrac{1-P_d}{1-P_{fa}}\log\dfrac{1-\pi_{1,1}}{1-\pi_{1,0}}>\dfrac{1-\pi_{1,1}}{1-\pi_{1,0}}+\dfrac{1-P_d}{1-P_{fa}}\log\dfrac{\pi_{1,1}}{\pi_{1,0}}.\label{log1}
\end{eqnarray} 
Since $P_d>P_{fa}$ and $t<0.5$, $\pi_{1,1}>\pi_{1,0}$. It can also be proved that $\dfrac{1-\pi_{1,1}}{1-\pi_{1,0}}>\dfrac{1-P_d}{1-P_{fa}}$. Hence, we have
\begin{small}
\begin{eqnarray}
&&
\dfrac{1-\pi_{1,1}}{1-\pi_{1,0}}>\dfrac{1-P_d}{1-P_{fa}}\left[1-(\pi_{1,1}-\pi_{1,0})\right]\\
&\Leftrightarrow&
\dfrac{1}{\pi_{1,1}-\pi_{1,0}}\left[\dfrac{\pi_{1,1}}{\pi_{1,0}}-\dfrac{1-\pi_{1,1}}{1-\pi_{1,0}}\right]>\dfrac{1-P_d}{1-P_{fa}}\left[\dfrac{1}{\pi_{1,0}}+\dfrac{1}{1-\pi_{1,1}}\right]\label{eq-7}\\
&\Leftrightarrow&
\dfrac{\pi_{1,1}}{\pi_{1,0}}-\dfrac{1-\pi_{1,1}}{1-\pi_{1,0}}>\dfrac{1-P_d}{1-P_{fa}}\left[\dfrac{\pi_{1,1}-\pi_{1,0}}{\pi_{1,0}}+\dfrac{\pi_{1,1}-\pi_{1,0}}{1-\pi_{1,1}}\right]\\
&\Leftrightarrow&
\dfrac{\pi_{1,1}}{\pi_{1,0}}+\dfrac{1-P_d}{1-P_{fa}}
\left[1-\dfrac{1-\pi_{1,0}}{1-\pi_{1,1}}\right]
>\dfrac{1-\pi_{1,1}}{1-\pi_{1,0}}
+\dfrac{1-P_d}{1-P_{fa}}\left[\dfrac{\pi_{1,1}}{\pi_{1,0}}-1\right].\label{eq-11}
\end{eqnarray}
\end{small}
Applying the logarithm inequality $(x-1)\geq\log x \geq \dfrac{x-1}{x}$, for $x>0$ to \eqref{eq-11}, one can prove that condition \eqref{log1} is true. Condition \eqref{log} and \eqref{log1} imply that any non zero deviation $\epsilon_i\in(0,p]$ in flipping probabilities $(P_{0,1}^{B}, P_{1,0}^{B})=(p-\epsilon_1,p-\epsilon_2)$ will result in an increase in the KLD.
\end{proof}

In the next theorem, we present a closed form expression for the optimal attacking probability distribution $(P_{j,1}^{B}, P_{j,0}^{B})$ that minimizes $KLD$ in the region where the attacker cannot blind the FC.
\begin{theorem}
\label{th3}
In the region where the attacker cannot blind the FC, the optimal attacking strategy is given by $(P_{0,1}^{B}, P_{1,0}^{B})=(1,1)$.
\end{theorem}
\begin{proof}
Observe that, in the region where the attacker cannot blind the FC, the optimal strategy comprises of symmetric flipping probabilities $(P_{0,1}^{B}= P_{1,0}^{B}=p)$. The proof is complete if we show that KLD, $D$, is a monotonically decreasing function of the flipping probability $p$.

Let us denote, $P(z=1|H_1)=\pi_{1,1}$ and $P(z=1|H_0)=\pi_{1,0}$. After plugging in $(P_{0,1}^{B}, P_{1,0}^{B})=(p,p)$ in \eqref{eq1} and \eqref{eq2}, we get
\begin{eqnarray}
\pi_{1,1}&=& t(p-P_d(2p))+P_d\\
\pi_{1,0}&=& t(p-P_{fa}(2p))+P_{fa}.
\end{eqnarray}
Now we show that the KLD, $D$, as given in \eqref{klde} is a monotonically  decreasing function of the parameter $p$ or in other words, $\dfrac{dD}{dp}<0$. After plugging in $\pi_{1,1}'=t(1-2P_d)$ and $\pi_{1,0}'=t(1-2P_{fa})$ in the expression of $\dfrac{dD}{dp}$ and rearranging the terms, the condition $\dfrac{dD}{dp}<0$ becomes 
\begin{eqnarray}
\label{log2}
(1-2P_{fa})\left(\dfrac{1-\pi_{1,1}}{1-\pi_{1,0}}-\dfrac{\pi_{1,1}}{\pi_{1,0}}\right)
+(1-2P_d)\log\left(\dfrac{1-\pi_{1,0}}{1-\pi_{1,1}}\dfrac{\pi_{1,1}}{\pi_{1,0}}\right)<0
\end{eqnarray}
Since $P_d>P_{fa}$ and $t<0.5$, we have $\pi_{1,1}>\pi_{1,0}$. Now, using the fact that $\dfrac{1-P_d}{1-P_{fa}}>\dfrac{1-2P_d}{1-2P_{fa}}$ and \eqref{eq-7}, we have
\begin{small} 
\begin{eqnarray}
&&
\dfrac{1}{\pi_{1,1}-\pi_{1,0}}\left[\dfrac{\pi_{1,1}}{\pi_{1,0}}-\dfrac{1-\pi_{1,1}}{1-\pi_{1,0}}\right]>\dfrac{1-2P_d}{1-2P_{fa}}\left[\dfrac{1}{\pi_{1,0}}+\dfrac{1}{1-\pi_{1,1}}\right]\\
&\Leftrightarrow&
\dfrac{\pi_{1,1}}{\pi_{1,0}}+\dfrac{1-2P_d}{1-2P_{fa}}
\left[1-\dfrac{1-\pi_{1,0}}{1-\pi_{1,1}}\right]
>\dfrac{1-\pi_{1,1}}{1-\pi_{1,0}}
+\dfrac{1-2P_d}{1-2P_{fa}}\left[\dfrac{\pi_{1,1}}{\pi_{1,0}}-1\right].\label{tin}
\end{eqnarray}
\end{small}
Applying the logarithm inequality $(x-1)\geq\log x \geq \dfrac{x-1}{x}$, for $x>0$ to~\eqref{tin}, one can prove that \eqref{log2} is true.
\end{proof} 

Next, to gain insights into the solution, we present some numerical results in Figure~\ref{dvsp} that corroborate our theoretical results.
We plot KLD as a function of the flipping probabilities $(P_{1,0}^B,\;P_{0,1}^B)$. We assume that the probability of detection is $P_d=0.8$, the probability of false alarm is $P_{fa}=0.2$ and the fraction of covered nodes by the Byzantines is $t=0.4$. It can be seen that the optimal attacking strategy comprises of symmetric flipping probabilities and is given by $(P_{0,1}^{B}, P_{1,0}^{B})=(1,1)$, which corroborate our theoretical result presented in Lemma~\ref{lem} and Theorem~\ref{th3}.

\begin{figure}[t]
  \centering
    \includegraphics[height=2.5in, width=!]{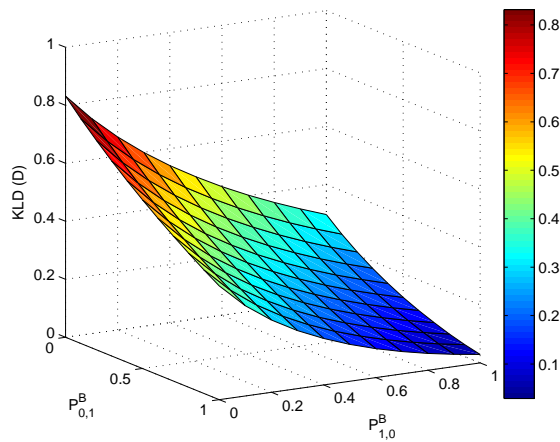}
    \vspace*{-0.1in}
    \caption{KL distance vs Flipping Probabilities when $P_d = 0.8$, 
$P_{fa} = 0.2$, and the fraction of covered nodes by the Byzantines is $t = 0.4$}\label{dvsp}
  \vspace*{-0.15in}
\end{figure}

Next, we explore some properties of the KLD with respect to the fraction of covered nodes $t$ in the region where the attacker cannot blind the FC, i.e., $t<0.5$.
\begin{lemma}
\label{equiv}
$D^*$ =$\underset{(P_{j,1}^{B}, P_{j,0}^{B})}{\text{min}} D(\pi_{j,1}||\pi_{j,0})$ is a continuous, decreasing and convex function of fraction of covered nodes by the Byzantines $t=\sum_{k=1}^{K}[\beta_{k}(\sum_{i=1}^{k}\alpha_{i})]$ in the region where the attacker cannot blind the FC ($t<0.5$).
\end{lemma}
\begin{proof}
The continuity of $D(\pi_{j,1}||\pi_{j,0})$ with respect to the involved distributions implies the continuity of $D^*$.
To show that $D^*$ is a decreasing function of $t$, we use the fact that $\underset{(P_{0,1}^{B}, P_{1,0}^{B})}{\text{argmin}} D(\pi_{j,1}||\pi_{j,0})$ is equal to $(1,1)$ for $t<0.5$ (as shown in Theorem~\ref{th3}). After plugging $(P_{0,1}^{B}, P_{1,0}^{B})=(1,1)$ in the KLD expression, it can be shown that the expression for the derivative of $D$ with respect to $t$, $\dfrac{dD}{dt}$, is the same as \eqref{log2}. Using the results of Theorem~\ref{th3}, it follows that $\dfrac{dD}{dt}<0$ and, therefore, $D^*$ is a monotonically decreasing function of $t$ in the region where $t<0.5$. The convexity of $D^*$ follows from the fact that $D^*(\pi_{j,1}||\pi_{j,0})$ is convex in $\pi_{j,1}$ and $\pi_{j,0}$, which are affine transformations of $t$ (Note that, convexity holds under affine transformation).
\end{proof}

\textit{It is worth noting that Lemma~\ref{equiv} suggests that by minimizing/maximizing the fraction of covered nodes $t$, the FC can maximize/minimize the KLD. Using this fact, from now onwards we will consider fraction of covered nodes $t$ in lieu of the KLD in further analysis in the paper.}

Next, to gain insights into the solution, we present some numerical results in Figure~\ref{dvst} that corroborate our theoretical results. We plot $\underset{(P_{j,1}^{B}, P_{j,0}^{B})}{\text{min}}$ KLD as a function of the fraction of covered nodes. We assume that the probabilities of detection and false alarm are $P_d=0.8$ and $P_{fa}=0.2$, respectively. Notice that, when $50\%$ of the nodes in the network are covered, KLD between the two probability distributions becomes zero and FC becomes blind. It can be seen that $D^*$ is a continuous, decreasing and convex function of the fraction of covered nodes $t$ in the region $t<0.5$, which corroborate our theoretical result presented in Lemma~\ref{equiv}.

\begin{figure}[t]
  \centering
    \includegraphics[height=2.5in, width=!]{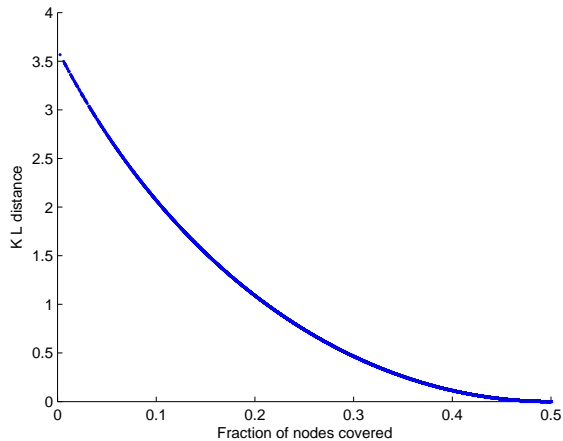}
    \vspace*{-0.1in}
    \caption{$\underset{(P_{j,1}^{B}, P_{j,0}^{B})}{\text{min}}$ KL distance vs Fraction of nodes covered when $P_d = 0.8$ and $P_{fa} = 0.2$}\label{dvst}
  \vspace*{-0.15in}
\end{figure}
Until now, we have explored the problem from the attacker's perspective. In the rest of the paper we look into the problem from a network designer's perspective and propose a technique to mitigate the effect of the Byzantines. More specifically, we explore the problem of designing a robust tree topology considering the Byzantine to incur a cost for attacking the network and the FC to incur a cost for deploying (including the cost of protection, etc.) the network. The FC (network designer) tries to design a perfect $a$-ary tree topology under its cost budget constraint such that the system performance metric, i.e., KLD is maximized. Byzantines, on the other hand, are interested in attacking or capturing nodes to cause maximal possible degradation in system performance, with the cost of attacking or capturing nodes not to exceed the attacker's budget. This problem can be formulated as a bi-level programming problem where the upper and the lower level problems with conflicting objectives belong to the leader (FC) and the follower (Byzantines), respectively.

\section{Robust Topology Design}
\label{sec4}
In this problem setting, it is assumed that there is a cost associated with attacking each node in the tree (which may represent resources required for capturing a node or cloning a node in some cases). We also assume that the costs for attacking nodes at different levels are different. Specifically, let $c_{k}$ be the cost of attacking any one node at level $k$. Also, we assume $c_{k}>c_{k+1}$ for $k= 1,\cdots,K-1$, i.e., it is more costly to attack nodes that are closer to the FC. Observe that, a node $i$ at level $k$ covers (in other words, can alter the decisions of)
all its successors and node $i$ itself. It is assumed that the network designer or the FC has a cost budget $C_{budget}^{network}$ and the attacker has a cost budget $C_{budget}^{attacker}$.
Let $P_{k}$ denote the number of nodes covered by a node at level $k$. We
refer to $P_{k}$ as the ``profit" of a node at level $k$.
Notice that,
$P_{k}=\frac{\sum_{i=k+1}^{K} N_{i}}{N_{k}}+1$.  

Notice that,
in a tree topology, $P_{k}$ can be written as
\begin{equation}
P_{k} = a_{k}\times P_{k+1}+1 \qquad for \;k= 1,...,K-1,
\label{eq17}
\end{equation}
where $P_{k}$ is the profit of attacking a node at level $k$, $P_{k+1}$ is the profit of attacking a node at level $k+1$ and $a_{k}$ is the number of immediate children of a node at level $k$. 
For a perfect $a$-ary tree $a_{k}=a,\;\forall k$ and $P_k=\frac{a^{K-k+1}-1}{a-1}$.
The FC designs the network, such that, given the attacker's budget, the fraction of covered nodes is minimized, and consequently a more robust perfect $a$-ary tree in terms of KLD (See Lemma~\ref{equiv}) is generated. Next, we formulate our robust topology design problem. 

\subsection{Robust Perfect $a$-ary Tree Topology Design}
Since the attacker aims to maximize the fraction of covered nodes by attacking/capturing $\{B_k\}_{k=1}^{K}$ nodes within the cost budget $C_{budget}^{attacker}$, the FC's objective is to minimize  the fraction of covered nodes by choosing the parameters $(K,\;a)$ optimally in a perfect $a$-ary tree topology $T(K,\;a)$ under its cost budget $C_{budget}^{network}$. This situation can be interpreted as a Bi-level optimization problem, where the first decision maker (the so-called leader) has the first choice, and the second
one (the so-called follower) reacts optimally to the leader's selection. It is the leader's aim to find such a decision which, together with the optimal response of the follower, optimizes the objective function of the leader. For our problem, the upper level problem (ULP) corresponds to the FC who is the leader of the game, while the lower level problem (LLP) belongs to the attacker who is the follower. We assume that the FC has complete information about the attacker's problem, i.e., the objective function and the constraints of the LLP. Similarly, the attacker is assumed to be aware about the FC's resources, i.e., cost of deploying the nodes $\{c_k\}_{k=1}^{K}$. Next, we formalize our robust perfect $a$-ary tree topology problem as follows:

\begin{equation}
\begin{split}
\underset{(K,\;a)\in\mathbb{Z}^{+}}{\mathrm{minimize}}\quad& \frac{\sum_{k=1}^{K}(a^{K-k+1}-1)B_{k}}{a(a^K-1)} \\
\mbox{subject to} \quad &  a_{min} \leq a \leq a_{max}\\
\quad &  K \geq K_{min}\\
\quad &  \sum_{k=1}^{K} a^k \geq N_{min}\\
\, & \sum_{k=1}^{K}c_{k}a^k \leq C_{budget}^{network}\\
\, &  \underset{B_{k}\in \mathbb{Z}^{+}}{\text{maximize}} \quad \frac{\sum_{k=1}^{K}(a^{K-k+1}-1)B_{k}}{a(a^K-1)} \\
\, &  \text{subject to}\quad \sum_{k=1}^{K}c_{k}B_{k} \leq C_{budget}^{attacker} \\
\, & \qquad \qquad \quad B_{k} \leq a^{k}, \forall\, k = 1,2,\ldots,K
\end{split}
\end{equation}
where $\mathbb{Z}^{+}$ is the set of non-negative integers, $a_{min}\geq 2$ and $K_{min}\geq 2$. The objective function in ULP is the fraction of covered nodes by the Byzantines $\frac{\sum_{k=1}^{K}P_{k}B_{k}}{\sum_{k=1}^{K}N_k}$, where $P_k=\frac{a^{K-k+1}-1}{a-1}$ and 
$\sum_{k=1}^{K}N_k=\frac{a(a^K-1)}{a-1}$. In the constraint $a_{min} \leq a \leq a_{max}$, $a_{max}$ represents the hardware constraint imposed by the Medium Access Control (MAC) scheme used and $a_{min}$ represents the design constraint enforced by the FC. 
The constraint on the number of nodes in the network $\sum_{k=1}^{K} a^k \geq N_{min}$ ensures that the network satisfies pre-specified detection performance guarantees. In other words, $N_{min}$ is the minimum number of nodes needed to guarantee a certain detection performance. The constraint on the cost expenditure $\sum_{k=1}^{K}c_{k}a^k \leq C_{budget}^{network}$ ensures that the
total expenditure of the network designer does not exceed the available budget.

In the LLP, the objective function is the same as that of the FC, but the sense of optimization is opposite, i.e., maximization of the fraction of covered nodes. 
The constraint $\sum_{k=1}^{K}c_{k}B_{k} \leq C_{budget}^{attacker}$ ensures that the total expenditure of the attacker does not exceed the available budget. The constraints $B_{k} \leq a^{k},\; \forall k$ are logical conditions, which prevent the attacker
from attacking non-existing resources.

Notice that, the bi-level optimization problem, in general, is an NP-hard problem~\cite{np}. In fact, the optimization problem corresponding to LLP is the packing formulation of the Bounded Knapsack Problem (BKP) \cite{Deineko}, 
which itself, in general, is NP-hard. Next, we discuss some properties of our objective function that enable our robust topology design problem to have a polynomial time solution.

\begin{lemma}
\label{lemk}
In a perfect $a$-ary tree topology, the fraction of covered nodes $\frac{\sum_{k=1}^{K}P_{k}B_{k}}{\sum_{k=1}^{K}N_k}$ by the attacker with the cost budget $C_{budget}^{attacker}$ for an optimal attack is a non-decreasing function of the number of levels $K$ in the tree.   
\end{lemma}
\begin{proof}
Let us denote the optimal attacking set for a $K$ level \textit{perfect $a$-ary} tree topology $T(K,\; a)$ by $\{B_k^1\}_{k=1}^{K}$ and the optimal attacking set for a \textit{perfect $a$-ary} tree topology with $K+1$ levels by $\{B_k^2\}_{k=1}^{K+1}$ given the cost budget $C_{budget}^{attacker}$. To prove the lemma, it is sufficient to show that
\begin{center}
\begin{equation}
\label{incK}
\frac{\sum_{k=1}^{K+1}P_{k}^{2}B_k^2}{\sum_{k=1}^{K+1}N_k}\geq\frac{\sum_{k=1}^{K}P_{k}^{2}B_{k}^1}{\sum_{k=1}^{K+1}N_k}\geq\frac{\sum_{k=1}^{K}P_{k}^{1}B_{k}^1}{\sum_{k=1}^{K}N_k},
\end{equation}
\end{center}
where $P_{k}^1$ is the profit of attacking a node at level $k$ in a $K$ level \textit{perfect $a$-ary} tree topology and $P_{k}^2$ is the profit of attacking a node at level $k$ in a $K+1$ level \textit{perfect $a$-ary} tree topology.

First inequality in \eqref{incK} follows due to the fact that $\{B_k^1\}_{k=1}^{K}$ may not be the optimal attacking set for topology $T(K+1,a)$. To prove the second inequality observe that, an increase in the value of parameter $K$ results in an increase in both the denominator (number of nodes in the network) and the numerator (fraction of covered nodes). Using this fact, let us denote
\begin{equation}
\dfrac{\sum_{k=1}^{K}P_{k}^{2}B_{k}^1}{\sum_{k=1}^{K+1}N_k}=\frac{x+x_1}{y+y_1}
\end{equation}
with $x={\sum_{k=1}^{K}P_{k}^{1}B_{k}^1}$ with $P_k^1=\dfrac{a^{K-k+1}-1}{a-1}$, $y={\sum_{k=1}^{K}N_k}=\dfrac{a(a^K-1)}{a-1}$, $x_1={\sum_{k=1}^{K}(B_{k}^{1} a^{K-k+1})}$ is the increase in the profit by adding one more level to the topology and $y_1={a^{K+1}}$ is the increase in the number of nodes in the network by adding one more level to the topology .

Note that $\dfrac{x+x_1}{y+y_1}>\dfrac{x}{y}$ 
if and only if \begin{equation}
\label{con}
\dfrac{x}{y}<\dfrac{x_1}{y_1},
\end{equation}
where $x,y,x_1,$ and $y_1$ are positive values. Hence, it is sufficient to prove that
\begin{center}
\begin{equation*}
\frac{a^{K+1}\sum_{k=1}^{K}\left(\frac{B_{k}^{1}}{a^{k}}\right)-\sum_{k=1}^{K}B_{k}^{1}}{a(a^K-1)}\leq \frac{\sum_{k=1}^{K}(B_{k}^{1} a^{K-k+1})}{a^{K+1}}.
\end{equation*}
\end{center} 
The above equation can be further simplified to
\begin{center}
\begin{equation*}
\sum_{k=1}^{K}\left(\frac{B_{k}^{1}}{a^{k}}\right)\leq \sum_{k=1}^{K}\left(\frac{B_{k}^{1}}{a}\right)
\end{equation*}
\end{center}
which is true for all $K\ge 1$.
\end{proof}

Next, to gain insights into the solution, we present some numerical results in Figure~\ref{tvsk} that corroborate our theoretical results. We plot the fraction of covered nodes by the Byzantines as a function of the total number of levels in the tree. We assume that $a=2$ and vary $K$ from $2$ to $9$. We also assume that the cost to attack nodes at different levels are given by $[c_1,\cdots,c_9]=[52,\;48,\;24,\;16,\;12,\;8,\;10,\;6,\;4]$ and the cost budget of the attacker is $C_{budget}^{attacker}=50$.
For each $T(K,\;2)$, we find the optimal attacking set $\{B_k\}_{k=1}^{K}$ by an exhaustive search.
It can be seen that the fraction of covered nodes is a non-decreasing function of the number of levels $K$, which corroborate our theoretical result presented in Lemma~\ref{lemk}.

\begin{figure}[t]
  \centering
    \includegraphics[height=2.5in, width=!]{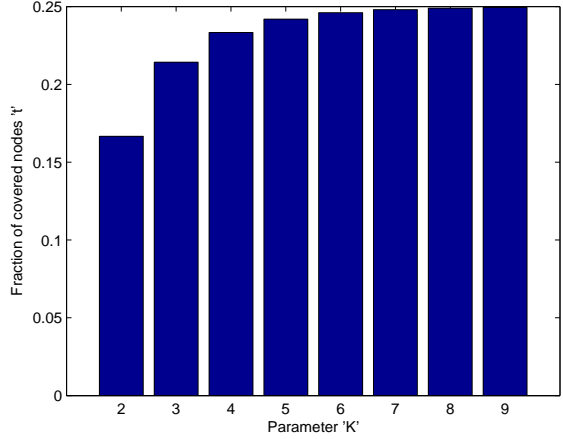}
    \vspace*{-0.1in}
    \caption{Fraction of nodes covered vs Parameter $K$ when $a = 2$, $K$ is varied
from 2 to 9, $[c_1,\cdots,c_9]=[52,\;48,\;24,\;16,\;12,\;8,\;10,\;6,\;4]$, and  $C_{budget}^{attacker}=50$}\label{tvsk}
  \vspace*{-0.15in}
\end{figure}

Next, we explore some properties of the fraction of covered nodes with parameter $a$ for a \textit{perfect $a$-ary} tree topology. Before discussing our result, we define the parameter $a_{min}$ as follows. For a fixed $K$ and attacker's cost budget $C_{budget}^{attacker}$, $a_{min}$ is defined as the minimum value of $a$ for which the attacker cannot blind the network or cover $50\%$ or more nodes. So we can restrict our analysis to $a_{min} \leq a\leq a_{max}$. Notice that, the attacker cannot blind all the trees $T(K,a)$ for which $a\geq a_{min}$ and can blind all the trees $T(K,a)$ for which $a<a_{min}$.\\
\begin{lemma}
\label{lema} 
In a \textit{perfect $a$-ary} tree topology, the fraction of covered nodes $\frac{\sum_{k=1}^{K}P_{k}B_{k}}{\sum_{k=1}^{K}N_k}$ by an attacker with cost budget $C_{budget}^{attacker}$ in an optimal attack is a decreasing function of parameter $a$ for a \textit{perfect $a$-ary} tree topology for $a\geq a_{min}\ge 2$.   
\end{lemma}
\begin{proof}
As before, let us denote the optimal attacking set for a $K$ level \textit{perfect $a$-ary} tree topology $T(K,\; a)$ by $\{B_k^1\}_{k=1}^{K}$ and the optimal attacking set for a \textit{perfect (a+1)-ary} tree topology $T(K,a+1)$  by $\{B_k^2\}_{k=1}^{K}$ given the cost budget $C_{budget}^{attacker}$. To prove the lemma, it is sufficient to show that 
\begin{center}
\begin{equation}
\label{inca}
\frac{\sum_{k=1}^{K}P_{k}^{2}B_k^2}{\sum_{k=1}^{K}N_k^2}<\frac{\sum_{k=1}^{K}P_{k}^{1}B_{k}^2}{\sum_{k=1}^{K}N_k^1}\leq\frac{\sum_{k=1}^{K}P_{k}^{1}B_{k}^1}{\sum_{k=1}^{K}N_k^1},
\end{equation}
\end{center}
where $N_k^1$ is the number of nodes at level $k$ in $T(K,a)$, $N_k^2$ is the number of nodes at level $k$ in $T(K,a+1)$, $P_{k}^1$ is the profit of attacking a node at level $k$ in $T(K,a)$ and $P_{k}^2$ is the profit of attacking a node at level $k$ in $T(K,a+1)$. Observe that, an interpretation of \eqref{inca} is that the attacker is using the attacking set $\{B_k^2\}_{k=1}^{K}$ to attack $T(K,\; a)$. However, one might suspect that  
the set $\{B_k^2\}_{k=1}^{k=K}$ is not a valid solution. More specifically, the set $\{B_k^2\}_{k=1}^{k=K}$ is not a valid solution in the following two cases:\\
1. \textit{$min(B_k^2,\;N_k^1)=N_k^1$ for any $k$:} For example, if $N_1^1=4$ for $T(K,4)$ and $B_1^2=5$ for $T(K,5)$ then it will not be possible for the attacker to attack $5$ nodes at level $1$ in $T(K,4)$ because the total number of nodes at level $1$ is $4$. In this case, $\{B_k^2\}_{k=1}^{K}$  might not be a valid attacking set for the tree $T(K,a)$.\\
2. \textit{$\{B_k^2\}_{k=1}^{k=K}$ is an overlapping set\footnote{We call $B_{k}$ and $B_{k+x}$ are overlapping, if the summation of $B_k^{k+x}$ and $B_{k+x}$ is greater than $N_{k+x}$, where $B_k^{k+x}$ is the number of nodes covered by the attacking set $B_{k}$ at level $k+x$. In a non-overlapping case, the attacker can always arrange nodes $\{B_{k}\}_{k=1}^{K}$ such that each path in the network has at most one Byzantine.} for $T(K,a)$:} For example, for $T(2,3)$ if $B_1^2=2$ and $B_2^2=4$, then, $B_1^2$ and $B_2^2$ are overlapping. In this case, $\{B_k^2\}_{k=1}^{K}$  might not be a valid attacking set for the tree $T(K,a)$.\\
However, both of the above conditions imply that the attacker can blind the network with $C_{budget}^{attacker}$ (See Appendix~\ref{ap2}), which cannot be true for $a\geq a_{min}$, and, therefore, $\{B_k^2\}_{k=1}^{K}$ will indeed be a valid solution. Therefore, \eqref{inca} is sufficient to prove the lemma.

Notice that, the second inequality in \eqref{inca} follows due to the fact that $\{B_k^2\}_{k=1}^{K}$ may not be the optimal attacking set for topology $T(K,a)$. 
To prove the first inequality in \eqref{inca}, we first 
consider the case where attacking set $\{B_k^2\}_{k=1}^{k=K}$ contains only one node, i.e., $B_k^2=1$ for some $k$, and show that $\frac{P_{k}^{2}}{\sum_{k=1}^{K}N_k^2}<\frac{P_{k}^{1}}{\sum_{k=1}^{K}N_k^1}$.
Substituting $P_k^1=\dfrac{a^{K-k+1}-1}{a-1}$ for some $k$ and  ${\sum_{k=1}^{K}N_k^1}=\dfrac{a(a^K-1)}{a-1}$ in the left side inequality of \eqref{inca}, we have 
\begin{equation*}
\frac{(a)^{K-k+1}-1}{(a)((a)^K-1)}> \frac{(a+1)^{K-k+1}-1}{(a+1)((a+1)^K-1)}.
\end{equation*}
After some simplification, the above condition becomes
\begin{eqnarray}
&&(a+1)^{K+1}[(a)^{K-k+1}-1]-(a)^{K+1}[(a+1)^{K-k+1}-1]\nonumber\\
&&+(a)[(a+1)^{K-k+1}-1]-(a+1)[(a)^{K-k+1}-1]> 0.\label{eq_a}
\end{eqnarray}
In  Appendix~\ref{ap1}, we show that
\begin{equation}
\label{ineq2}
(a)[(a+1)^{K-k+1}-1]-(a+1)[(a)^{K-k+1}-1]> 0
\end{equation}
and
\begin{equation}
\label{ineq1}
(a+1)^{K+1}[(a)^{K-k+1}-1]-(a)^{K+1}[(a+1)^{K-k+1}-1]\geq 0.
\end{equation}
From \eqref{ineq1} and \eqref{ineq2}, condition~\eqref{eq_a} holds.

Since we have proved that
$$\frac{P_{k}^{2}}{\sum_{k=1}^{K}N_k^2}<\frac{P_{k}^{1}}{\sum_{k=1}^{K}N_k^1}\mbox{ for all } 1\le k\le K,$$
to generalize the proof for any arbitrary attacking set $\{B_k^2\}_{k=1}^{K}$ we multiply both sides of the above inequality with $B^2_k$ and sum it over all $1\leq k\leq K$ inequalities. Now, we have 
$$\frac{\sum_{k=1}^{K}P_{k}^{2}B_k^2}{\sum_{k=1}^{K}N_k^2}<\frac{\sum_{k=1}^{K}P_{k}^{1}B_{k}^2}{\sum_{k=1}^{K}N_k^1}.$$

\end{proof}

\begin{figure}[t]
  \centering
    \includegraphics[height=2.5in, width=!]{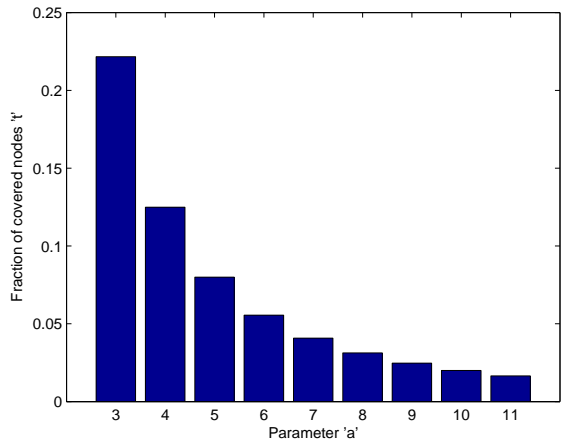}
    \vspace*{-0.1in}
    \caption{Fraction of nodes covered vs Parameter $a$ when $K=6$, parameter $a$ is varied
    from 3 to 11, $[c_1,\cdots,c_9]=[52,\;48,\;24,\;16,\;12,\;8,\;10,\;6,\;4]$, and  $C_{budget}^{attacker}=50$}\label{tvsa}
  \vspace*{-0.15in}
\end{figure}

Next, to gain insights into the solution, we present some numerical results in Figure~\ref{tvsa} that corroborate our theoretical results.
We plot the fraction of covered nodes by the Byzantines as a function of the parameter $a$ in the tree. We assume that the parameter $K=6$ and vary $a$ from $3$ to $11$. We also assume that the cost to attack nodes at different levels are given by $[c_1,\cdots,c_9]=[52,\;48,\;24,\;16,\;12,\;8,\;10,\;6,\;4]$ and the cost budget of the attacker is $C_{budget}^{attacker}=50$.
For each $T(6,\;a)$ we find the optimal attacking set $\{B_k\}_{k=1}^{K}$ by an exhaustive search.
It can be seen that the fraction of covered nodes is a decreasing function of the parameter $a$, which corroborate our theoretical result presented in Lemma~\ref{lema}.

Next, based on the above Lemmas we present an algorithm which can solve our robust perfect $a$-ary tree topology design problem (bi-level programming problem) efficiently.\\

\subsection{Algorithm for solving Robust Perfect $a$-ary Tree Topology Design Problem}

\begin{algorithm} [] 
\small                       
\caption{Robust Perfect $a$-ary Tree Topology Design}          
\label{robust}                           
\begin{algorithmic} [1]
\REQUIRE $c_{k} > c_{k+1} \qquad for \;k= 1,...,K-1$
\STATE $K\leftarrow K_{min}$; $a\leftarrow a_{max}$

\IF{$\left(\sum_{k=1}^{K} c_k a^{k}> C_{budget}^{network}\right)$}\label{if-loop}
\STATE Find the largest integer $a-\ell$, $\ell\ge 0$, such that $\sum_{k=1}^{K} c_k (a-\ell)^{k}\le C_{budget}^{network}$\label{find}
\IF{$\left(a-\ell<a_{min}\right)$}
\STATE \textbf{return} $(\phi,\phi)$
\ELSE
\STATE 
$a\leftarrow a-\ell$
\ENDIF
\ENDIF
\IF{$\left(\sum_{k=1}^{K} a^{k}\geq N_{min}\right)$}
\STATE \textbf{return} \text{$(K, a)$}
\ELSE
\STATE $K\leftarrow K+1$
\STATE \textbf{return to} \text{Step}~\ref{if-loop}
\ENDIF
\end{algorithmic}
\end{algorithm}
Based on Lemma \ref{lemk} and Lemma \ref{lema}, we  present a polynomial time algorithm for solving the robust perfect $a$-ary tree topology design problem.
Observe that, the robust network design problem is equivalent to designing perfect $a$-ary tree topology with minimum $K$ and maximum $a$ that satisfy network designer's constraints. In Algorithm~\ref{robust}, we start from the solution candidate $(a_{max},\; K_{min})$. If it does not satisfy the cost expenditure constraint we reduce $a_{max}$ by one, i.e., $a_{max}\leftarrow a_{max}-1$. Next, the algorithm checks for the total number of nodes constraint and if it is not satisfied, we increase $K_{min}$ by one, i.e., $K_{min}\leftarrow K_{min}+1$. After these steps, the algorithm checks whether this new solution candidate satisfies both the constraints. If it does, this will be the solution for the problem, otherwise, the algorithm solves the problem recursively until the hardware constraint is violated, 
i.e., $a<a_{min}$. In this case ($a<a_{min}$), we will not have any feasible solution which satisfies the network designer's constraints.\\  

This procedure greatly reduces the complexity because we do not need to solve the lower level problem in this case. Next, we prove that Algorithm~\ref{robust} indeed yields an optimal solution.
\begin{lemma}
Robust Perfect $a$-ary Tree Topology Design algorithm (Algorithm~\ref{robust}) yields an optimal solution $(K^*,\; a^*)$, if one exists.
\end{lemma}
\begin{proof}
Assume that the optimal solution exists. Let us denote by $(K^*,\; a^*)$, the optimal solution given by Algorithm~\ref{robust}. The main idea behind our proof is that any solution $(K,a)$ with $K\geq K^*$ and $a\leq a^*$ cannot perform better than $(K^*,\; a^*)$ as suggested by Lemma \ref{lemk} and Lemma \ref{lema}. By transitive property, it can be proved that any solution $(K,a)$ with $K\geq K^*$ and  $a\leq a^*$ cannot perform better than $(K^*,\; a^*)$.
Also, observe that, the only feasible solution in the region ($K_{min}\leq K\leq K^*,\;a^*\leq a\leq a_{max}$) is $(K^*,\; a^*)$. This implies that $(K^*,\; a^*)$ is an optimal solution.

Notice that, our algorithm searches for the feasible solution with the smallest $K$ and the largest $a$. Any feasible solution $(K,\ a)$ satisfies the following two conditions:
\begin{enumerate}
\item $\sum_{k=1}^{K} c_k a^{k}\le C_{budget}^{network}$;\label{condition-1}
\item $\sum_{k=1}^{K} a^{k}\geq N_{min}$.\label{condition-2}
\end{enumerate}

By Lemma~\ref{lema}, if $(K,a)$ is a feasible solution, then $(K,a')$ with $a'<a$ will not be a better solution than $(K,a)$. Hence, for a given $K$, Step~\ref{find} only locates the  solution with largest $a$ for a given $K$. Furthermore, if both $(K,a)$ and $(K',a')$ satisfy  Condition~\ref{condition-1} and $K<K'$, then $a\ge a'$. Hence, for a given $K$, the largest $a$ in the current iteration satisfying Condition~\ref{condition-1} cannot be larger than the $a$ found in the previous iteration. This verifies that $\ell\ge 0$ is a sufficient condition to find the largest $a$ in Step~\ref{find}.

Next, we prove that Algorithm~\ref{robust} can stop when the first feasible solution has been found. Let $(K^1,\; a^1)$ be the first feasible solution found by Algorithm~\ref{robust}. It is clear that the next feasible solution $(K,a)$ must have $K>K^1$ and $a\le a^1$, since, the algorithm increases $K$ and it satisfies Condition~\ref{condition-1}. Algorithm~\ref{robust} stops 
when both Condition~\ref{condition-1} and Condition~\ref{condition-2} satisfy.

By the previous argument given in the beginning of the proof, we conclude that $(K,a)$ does not perform better than $(K^1,\; a^1)$. Hence, $(K^1,\; a^1)$ is the optimal solution $(K^*,\; a^*)$. It can be seen that if there is no solution, then the algorithm will return $(\emptyset,\emptyset)$. This is due to the fact that if $a-\ell<a_{min}$, then no $a$ can satisfy Condition~\ref{condition-1} for current and further iterations. Hence, the algorithm terminates and returns $(\emptyset,\emptyset)$.
\end{proof}

\begin{figure}[t]
  \centering
    \includegraphics[height=2.5in, width=!]{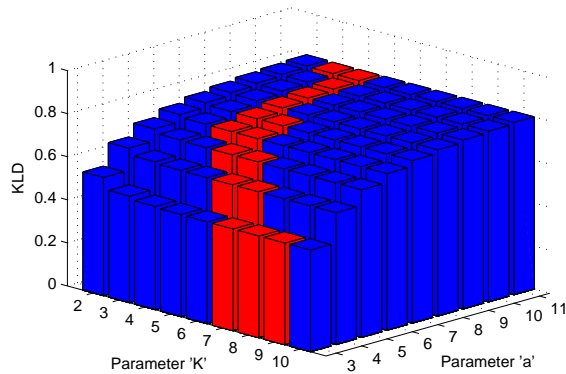}
    \vspace*{-0.1in}
    \caption{KLD vs Parameters 'K' and 'a' when $(P_d,P_{fa}) = (0.8,0.2)$, $C_{budget}^{network}=400000$, $C_{budget}^{attacker}=50$ and
$N_{min}=1400$}\label{opt}
  \vspace*{-0.15in}
\end{figure}
Next, to gain insights into the solution, we present some numerical results in Figure~\ref{opt} that corroborate our theoretical results.
We plot the $\underset{P_{1,0},P_{0,1}}{\text{min}}$ KLD for all the combinations of parameter $K$ and $a$ in the tree. We vary the parameter $K$ from $2$ to $10$ and $a$ from $3$ to $11$. We also assume that the costs to attack nodes at different levels are given by $[c_1,\cdots,c_{10}]=[52,\;50,\;25,\;24,\;16,\;10,\;8,\;6,\;5,\;4]$, and cost budgets of the network and the attacker are given by $C_{budget}^{network}=400000$,
$C_{budget}^{attacker}=50$, respectively. The node budget constraint is assumed to be $N_{min}=1400$.
For each $T(K,\;a)$, we find the optimal attacking set $\{B_k\}_{k=1}^{K}$ by an exhaustive search. All the feasible solutions are plotted in red and unfeasible solutions are plotted in blue. Notice that, $T(K_{min},\;a_{max})$ which is $T(2,\;11)$ is not a feasible solution and, therefore, if we use Algorithm~\ref{robust} it will try to find the feasible solution which has minimum possible deviation from $T(K_{min},\;a_{max})$. It can be seen that the optimal solution $T(3,\;11)$ has minimum possible deviation from the $T(K_{min},\;a_{max})$, which corroborate our algorithm.

%
\section{Conclusion}
\label{sec6}
In this paper, we have considered distributed
detection in perfect $a$-ary tree topologies in the presence of Byzantines, and characterized the power of attack analytically. We provided closed-form expressions for minimum attacking power required by the Byzantines to blind the FC. We obtained closed form expressions for the optimal attacking strategies that minimize the detection error exponent at the FC. We also looked at the possible counter-measures from the FC's perspective to protect the network from these Byzantines. We formulated the robust
topology design problem as a bi-level program and provided an efficient algorithm to solve it. 
There are still many interesting questions that
remain to be explored in the future work such as an analysis of the problem for arbitrary topologies.
Note that, some analytical methodologies used in this paper are certainly exploitable for studying the attacks in different topologies. Other questions such as the case
where Byzantines collude in several groups (collaborate) to degrade the detection performance
can also be investigated.
\section*{Acknowledgement}
This work was supported in part by ARO under Grant W911NF-09-1-0244 and AFOSR under Grants FA9550-10-1-0458, FA9550-10-1-0263.

\appendices
\section{}
\label{ap2}
We want to show that the set $\{B_k\}_{k=1}^{K}$ can blind the FC if any of following two cases is true.\\
1. $min(B_k,\;N_k)=N_k$ for any $k$,\\
2. $\{B_k\}_{k=1}^{k=K}$ is an overlapping set\\
In other words, set $\{B_k\}_{k=1}^{K}$ covers $50\%$ or more nodes. Let us denote by $\tilde{k}$, the $k$ for which  $min(B_k,\;N_k)=N_k$ (there can be multiple such $k$). Then $\{B_k\}_{k=1}^{K}$ satisfies 
\begin{center}
\begin{equation}
\label{min}
\frac{\sum_{k=1}^{K}P_{k}B_k}{\sum_{k=1}^{K}N_k}\geq\frac{P_{\tilde{k}}B_{\tilde{k}}}{\sum_{k=1}^{K}N_k}\geq\frac{P_{\tilde{k}}N_{\tilde{k}}}{\sum_{k=1}^{K}N_k}\geq\frac{P_{K}N_{K}}{\sum_{k=1}^{K}N_k}.
\end{equation}
\end{center}
Similarly, let us assume $B_{k'}$ and $B_{\tilde{k}}$ are overlapping with $\tilde{k}=k'+x$ (there can be multiple overlapping $k$). Then $\{B_k\}_{k=1}^{K}$ satisfies
\begin{center}
\begin{equation}
\label{over}
\frac{\sum_{k=1}^{K}P_{k}B_k}{\sum_{k=1}^{K}N_k}\geq\frac{P_{\tilde{k}}B_{\tilde{k}}+P_{k'}B_{k'}}{\sum_{k=1}^{K}N_k}\geq\frac{P_{\tilde{k}}N_{\tilde{k}}}{\sum_{k=1}^{K}N_k}\geq\frac{P_{K}N_{K}}{\sum_{k=1}^{K}N_k}.
\end{equation}
\end{center}
Observe that, to prove our claim it is sufficient to show that
\begin{equation}
\label{pro}
\frac{P_{K}N_{K}}{\sum_{k=1}^{K}N_k}\geq 0.5 \Leftrightarrow P_{K}N_{K}\geq \frac{N}{2}.
\end{equation}
Using the fact that for a Perfect $a$-ary tree $P_K=1$, $N_K=a^K$ and $N=\frac{a(a^K-1)}{a-1}$ the condition~\eqref{pro} becomes
\begin{equation}
\label{eq-app3}
2\times a^K\geq \frac{a(a^K-1)}{a-1}. 
\end{equation}
When $a\geq 2$, we have
\begin{eqnarray*}
&&
a \times a^K \geq 2\times a^K\\
&\Leftrightarrow&
a+a^{K+1} \geq 2\times a^{K}\\
&\Leftrightarrow&
2\times a^{K+1}-2\times a^K\geq a^{K+1}-a\\
&\Leftrightarrow&
2\times a^K\geq \frac{a(a^K-1)}{a-1}.
\end{eqnarray*}
Hence, \eqref{pro} holds and this completes our proof.
\section{}
\label{ap1}
We skip the proof of \eqref{ineq2} and only focus on the proof of \eqref{ineq1}. 
To show 
$$(a+1)^{K+1}[(a)^{K-k+1}-1]-(a)^{K+1}[(a+1)^{K-k+1}-1]\geq 0\mbox{ for $a\ge 2$}$$
is equivalent to show
$$a^{K+1}[(a-1)^{K-k+1}-1]-(a-1)^{K+1}[a^{K-k+1}-1]\geq 0\mbox{ for $a\ge 3$}$$
which can be simplified to
\begin{equation}
\label{eq-app}
(a(a-1))^{K-k+1}[a^k-(a-1)^k]\geq [a^{K+1}-(a-1)^{K+1}].
\end{equation}
Using binomial expansion, \eqref{eq-app} becomes
\begin{footnotesize} 
\begin{eqnarray}
&&
(a(a-1))^{K-k+1}[a^{k-1}+(a-1)a^{k-2}+\cdots+(a-1)^{k-1}]\geq\nonumber\\
& &
[a^{K}+(a-1)a^{K-1}+\cdots+(a-1)^{K-1}a+(a-1)^{K}]\nonumber\\
&\Leftrightarrow &
\underbrace{(a-1)^{K-k+1}[a^{K}+(a-1)a^{K-1}+\cdots+(a-1)^{k-1}a^{K-k+1}]}_{\text{k terms}}\geq\nonumber\\
& &
\underbrace{[a^{K}+(a-1)a^{K-1}+\cdots+(a-1)^{k-1}a^{K-k+1}]}_{\text{k terms}} +\nonumber\\
& &
\underbrace{[(a-1)^{k}a^{K-k}+\cdots+(a-1)^{K-1}a+(a-1)^{K}]}_{\text{K-k+1 terms}}\nonumber\\
&\Leftrightarrow &
((a-1)^{K-k+1}-1)[a^{K}+\cdots+(a-1)^{k-1}a^{K-k+1}]\geq\nonumber\\
& &
[(a-1)^{k}a^{K-k}+\cdots+(a-1)^{K-1}a+(a-1)^{K}].\label{eq-app2}
\end{eqnarray}
\end{footnotesize}
Since $a\ge 3$, we have $((a-1)^{K-k+1}-1)\geq (K-k+1)\geq 1$. Hence,
\begin{footnotesize}
\begin{equation*}
((a-1)^{K-k+1}-1)[a^{K}+\cdots+(a-1)^{k-1}a^{K-k+1}]\geq
\end{equation*}
\begin{equation}
\label{inequ}
((a-1)^{K-k+1}-1)a^{K}\geq \underbrace{[(a-1)^{k}a^{K-k}+\cdots+(a-1)^{K}]}_{\text{K-k+1 terms}}
\end{equation}
\end{footnotesize}
and \eqref{eq-app2} holds.

\bibliographystyle{IEEEtran}
\bibliography{Conf,Book,Journal}
\end{document}